\documentclass[a4paper,UKenglish,cleveref, autoref, thm-restate]{lipics-v2021}

\pdfoutput=1 
\hideLIPIcs  

\usepackage{amsmath,amssymb,amsthm,mathtools}
\usepackage[small]{complexity}
\usepackage{xspace}
\usepackage{colortbl}
\usepackage[dvipsnames]{xcolor}

\usepackage{ebproof}
\usepackage{graphicx}

\usepackage{makecell}

\usepackage{tikz}
\usetikzlibrary{automata, arrows.meta, positioning, shapes, quotes, cd}

\tikzset{> = {Stealth[length=2.9mm, width=1.25mm, inset=0.48mm]}}

\DeclarePairedDelimiter\set{\lbrace}{\rbrace}
\DeclarePairedDelimiter\abs{\lvert}{\rvert}

\DeclarePairedDelimiter\parens{(}{)}
\DeclarePairedDelimiter\tuple{\langle}{\rangle}

\DeclarePairedDelimiterX{\map}[2]{\lbrace}{\rbrace}{#1 \,\delimsize\vert\,\mathopen{} #2}
\DeclarePairedDelimiterX{\cond}[2]{[}{]}{#1 \,\delimsize\vert\,\mathopen{} #2}
\DeclarePairedDelimiterX{\subrange}[2]{[}{]}{#1 :\mathopen{} #2}
\DeclarePairedDelimiterX{\range}[2]{[}{]}{#1 ..\mathopen{} #2}

\newcommand{\BigO}{\mathcal{O}}

\renewcommand{\phi}{\varphi}

\title{The Power of Small Symmetries} 

\author{Nikita Gaevoy}
{The Faculty of Data and Decision Sciences, Technion, Israel}
{nikgaevoy@gmail.com}{}{}

\authorrunning{N. Gaevoy} 

\Copyright{Nikita Gaevoy} 

\ccsdesc[500]{Theory of computation~Proof complexity}

\keywords{proof complexity, complexity lower bounds, resolution with symmetries, small symmetries} 

\nolinenumbers 

\EventEditors{John Q. Open and Joan R. Access}
\EventNoEds{2}
\EventLongTitle{42nd Conference on Very Important Topics (CVIT 2016)}
\EventShortTitle{CVIT 2016}
\EventAcronym{CVIT}
\EventYear{2016}
\EventDate{December 24--27, 2016}
\EventLocation{Little Whinging, United Kingdom}
\EventLogo{}
\SeriesVolume{42}
\ArticleNo{23}

\newlang{\bphp}{bPHP}
\newlang{\ofphp}{ofPHP}
\newlang{\gop}{GOP}

\usepackage{xparse}

\NewDocumentCommand{\SRCI}{o}{
    \IfNoValueTF{#1}
    {\text{SRC-I}}
    {{#1}\text{-SRC-I}}%
}
\NewDocumentCommand{\SRCII}{o}{%
    \IfNoValueTF{#1}
    {\text{SRC-II}}
    {{#1}\text{-SRC-II}}%
}
\NewDocumentCommand{\SRCIII}{o}{%
    \IfNoValueTF{#1}%
    {\text{SRC-III}}%
    {{#1}\text{-SRC-III}}%
}
\NewDocumentCommand{\SRI}{o}{%
    \IfNoValueTF{#1}%
    {\text{SR-I}}%
    {{#1}\text{-SR-I}}%
}
\NewDocumentCommand{\SRII}{o}{%
    \IfNoValueTF{#1}%
    {\text{SR-II}}%
    {{#1}\text{-SR-II}}%
}
\NewDocumentCommand{\SRIII}{o}{%
    \IfNoValueTF{#1}%
    {\text{SR-III}}%
    {{#1}\text{-SR-III}}%
}
\NewDocumentCommand{\HRI}{o}{%
    \IfNoValueTF{#1}%
    {\text{HR-I}}%
    {{#1}\text{-HR-I}}%
}
\NewDocumentCommand{\HRII}{o}{%
    \IfNoValueTF{#1}%
    {\text{HR-II}}%
    {{#1}\text{-HR-II}}%
}
\NewDocumentCommand{\HRIII}{o}{%
    \IfNoValueTF{#1}%
    {\text{HR-III}}%
    {{#1}\text{-HR-III}}%
}
\newcommand{\Res}{\text{Res}\xspace}

\begin{document}

\maketitle


\begin{abstract}
Resolution with symmetries is a natural extension of the Resolution proof system that allows to use symmetries of the formula to simplify the proof.
Symmetries can be global (applied to the whole input formula), local (applied to a subformula), or dynamic (applied to newly derived clauses as well).
The framework of Resolution with (global) symmetries was introduced by Krishnamurthy~\cite{Krishnamurthy85} and further extended by Arai and Urquhart~\cite{AU00} to local symmetries.
Later, Szeider~\cite{Szeider05} generalized this approach to homomorphisms and introduced the notion of Resolution with dynamic symmetries.
While proving superpolynomial proof-size lower bounds for Resolution with dynamic symmetries remains an open problem already for two decades, the power of proof systems with global and local symmetries is well studied: exponential lower bounds have been proven for these proof systems, as well as exponential separations between all of them.
However, these systems are too general to reflect practical applications since it is computationally too hard to find and efficiently exploit arbitrary symmetries.

In this work, we introduce the notion of small symmetries: symmetries that can operate on a limited number of variables at the same time.
Resolution with small symmetries gives hopes both for practical applications and for theoretical study of dynamic symmetries.
We show that proof systems with both local and global small symmetries form strict hierarchies w.r.t. the size of symmetries.
We prove exponential separations between proof systems with symmetries of different sizes and types.
It turns out that even lower levels of these hierarchies are exponentially separated from Resolution and stronger proof systems, such as constant-depth Frege.
As a byproduct of our constructions, we obtain an exponential separation between the classical systems $\SRCI$ and $\SRII$ that was not known before.
\end{abstract}


\section{Introduction}\label{sec:introduction}

Many families of formulas based on combinatorial principles that are hard for Resolution, including pigeonhole and Tseitin formulas~\cite{BSW01,Haken85,Urquhart87}, have a very symmetric nature.
One of the first theoretical approaches to use these symmetries to optimize the length of Resolution-like proofs was Krishnamurthy's symmetry rule~\cite{Krishnamurthy85}.
This rule allows one to infer a clause symmetric to a previously deduced clause in one step, if the formula itself is symmetric under the same symmetry.
Resolution with this additional rule is called Resolution with global ($\SRI$) or local ($\SRII$) symmetries, depending on whether the whole input formula or only a part must be symmetric.
In both cases, symmetries are applied only to the original formula (not inferred clauses), so we call them \emph{static}.
Urquhart~\cite{Urquhart99} generalized these proof systems by introducing Resolution with \emph{complementary} symmetries ($\SRCI$ and $\SRCII$) that additionally allowed to interchange literals with their negations.

These systems were further extended by Szeider~\cite{Szeider05} in two important directions: Resolution with homomorphisms ($\HRI$ and $\HRII$) and Resolution with \textit{dynamic} symmetries ($\SRIII$, $\SRCIII$, and $\HRIII$).
Resolution with homomorphisms allows the usage of any homomorphism of a formula in place of a symmetry, which makes $\HRI$ and $\HRII$ even more general than their SR and SRC counterparts.
Resolution with dynamic symmetries (homomorphisms) is an extension that allows using inferred clauses in the symmetry rule.
This yields three new proof systems: $\SRIII$, $\SRCIII$, and $\HRIII$, depending on whether we are allowed to use symmetries, complementary symmetries, or homomorphisms.

\subsection{Previous work}

For all six proof systems with global and local symmetries and homomorphisms, exponential lower bounds on the length of their proofs are known (see Figure~\ref{fig:systems}).
Moreover, for Resolution with local homomorphisms ($\HRII$), the strongest proof system among these six, Szeider~\cite{Szeider05} showed that every formula could be transformed so that the lower bounds for the original formula in Resolution could be translated into lower bounds for the new version of the formula in $\HRII$.
However, it remains an open problem to prove superpolynomial lower bounds on proof length even in $\SRIII$, the weakest of the three proof systems with dynamic symmetries.
In fact, Resolution with dynamic symmetries is so much stronger than Resolution with static (i.e.,\ global or local) symmetries that, for some formulas, it is interesting whether using static symmetries only is sufficient to achieve polynomial proof length~\cite{Schweitzer21}.

\begin{figure}
    \centering

    \begin{tabular}{|l|c|c|c|}
        \hline
        & global & local & dynamic \\
        \hline
        symmetries & $\SRI$~\cite{Urquhart99} & $\SRII$~\cite{AU00} & $\SRIII$~(open) \\
        \hline
        \makecell[l]{symmetries + \\ complementation}& $\SRCI$~\cite{Urquhart99} & $\SRCII$~\cite{AU00} & $\SRCIII$~(open) \\
        \hline
        homomorphisms & $\HRI$~\cite{Szeider05} & $\HRII$~\cite{Szeider05} & $\HRIII$~(open) \\
        \hline
    \end{tabular}

    \caption{Resolutions with symmetries of varying degree of generality and corresponding exponential lower bounds. The strength of proof systems increases from top to bottom and from left to right.}
    \label{fig:systems}
\end{figure}

Despite the theoretical nature of this approach, some restricted versions of these proof systems could be found in practical SAT solvers.
Several such methods are jointly called dynamic symmetry handling~\cite{DBB17} or dynamic symmetry exploitation~\cite{Strichman21}, as opposed to symmetry breaking, the collection of methods that can add clauses to a formula that may affect the set of satisfying assignments as long as the satisfiability stays intact.
Symmetry breaking approaches also could be formalized and are promising for studying in the framework of proof systems~\cite{Bogaerts22}.
Although most solvers that use dynamic symmetry handling lose the competition to solvers with modern symmetry-breaking heuristics~\cite{DBB17}, the latter approach requires the symmetries to be essentially global.
In contrast, dynamic symmetry handling heuristics may efficiently work with local symmetries, which is crucial for almost symmetric formulas.
Moreover, in some practical cases, such as problems that emerge from bounded model checking, local symmetries arise naturally and could be generated together with the input instance of the satisfiability problem.
It allows the solver to have a sufficient amount of symmetries to work with without spending much time finding those symmetries, which can significantly improve the performance~\cite{Strichman21}.
In Ramsey theory, many state-of-the-art results are obtained using SAT solvers~\cite{Heule18,Heule16,Konev15,Kouril08}.
Although these solvers frequently employ symmetry breaking, the input formulas naturally contain many local symmetries, suggesting that their use could be a promising direction for further improvement.

\subsection{Our contribution}

In this work, we introduce a notion of small symmetries, that is, symmetries operating only on a small number (this number is a parameter of the proof system) of variables at the same time.
This defines a set of new proof systems that play an intermediate role between Resolution and the original proof systems with symmetries and homomorphisms.
A system with $k$-symmetries (where $k$ could be a constant or a function) is obtained from a system with arbitrary symmetries by restricting all applications of the symmetry rule to symmetries that nontrivially permute at most $k$ variables at a time (for formal definitions refer to Section~\ref{sec:definitions}).
We use the notation $k$-$\Pi$ to denote the proof system $\Pi$ with the restriction on $k$-symmetries (e.g., $\SRI[k]$, $\SRCII[k]$).
The notion of small symmetries arises naturally in the context of practical SAT solving, where it is much harder to detect large symmetries than small ones.
Additionally, many formulas have very symmetric nature, with many small symmetries, but almost no large symmetries.
Examples include formulas constructed from the bounded model checking instances and formulas encoding statements from Ramsey theory.
In the latter case, one can often find symmetries comparable in the size to the forbidden structure, which is often much smaller than the number of variables.

For dynamic symmetries, the notion of small symmetries allows us to define natural restricted versions of the proof system, which could be easier to prove lower bounds for.
It is important to note that even the most restricted system $\SRCIII[1]$ ($\SRIII[1]$ is trivially equivalent to Resolution) is probably strictly stronger than Resolution, even though we don't have any nontrivial upper or lower bounds on it.
In all three cases of global, local, and dynamic symmetries, it is natural to ask how the power of Resolution with small symmetries depends on the number of variables that the symmetries are allowed to operate on.

We focus on the case of small static symmetries.
We consider two different families of proof systems: Resolution with constant-size symmetries and Resolution with logarithmic-size symmetries.
In the case of logarithmic symmetries, we prove that $\SRI[(2 \log n)]$ has a polynomial proof for the formula $\bphp_n$ representing the binary encoding of the pigeonhole principle, which is hard for Resolution.
We also use this bound to show that $\SRI[(k + 1)]$ has exponentially shorter proofs for some formulas than $\SRCII[k]$, where $k$ is odd and at least logarithmic in $n$.
In the case of constant symmetries, we use Tseitin formulas to obtain an exponential separation between $\SRCI[3]$ and $\Res$ to show that for any $k$, $\SRI[\BigO(k)]$ has exponentially shorter proofs for some formulas than $\SRCII[k]$ (see Figure~\ref{fig:diagram_tseitin_triangular} for the precise relations between proof systems).

The hierarchy of proof systems with small symmetries is three-dimensional: we may increase the size of symmetries, add complementary symmetries, or generalize the type of used symmetries.
Therefore, there are three different types of separations that we can obtain, showing that the increase in one parameter makes the system stronger on some formulas, even if the other two decrease.
In our case, we use the separations that increase the size of symmetries to obtain the results of the other two types, and show that these separations hold even if the resulting system is allowed to operate with constant symmetries only (see Figure~\ref{fig:diagram_comb} for the diagram).
As a byproduct of these results, we also obtain an exponential separation between classical systems $\SRCI$ and $\SRII$ that was not known before (see Corollary~\ref{cor:last2}).

\subsection{Our methods}

All previous separations between proof systems with symmetries share the same structure: take a formula with a symmetric structure and then modify it to break one type of symmetries, but maintain another.
The base formulas in these separations are the propositional pigeonhole principle or perfect matching~\cite{Szeider05,Urquhart99}, or formulas that use them as a gadget~\cite{AU00}.
Unfortunately, these formulas allow one to construct proofs with large (polynomial-size) symmetries only.
The second problem arises from the fact that, although they allow using very broad families of formulas, the transformations produce only formulas with large symmetries.

Our separations follow the same path.
However, to overcome the limitations of the previous separations, we consider different families of formulas.
For each family, we construct polynomial-size refutations using only small symmetries.
We then apply transformations designed specifically for these formula types, which allow us to ``lift'' the lower bounds for the Resolution proof system without a significant growth in the size of the required symmetries.

\subparagraph*{Upper bounds.}
Our basic formulas are the binary encoding of the pigeonhole principle and Tseitin formulas.
For both families of formulas we construct refutations that have a dynamic programming-like nature.

For the refutation of $\bphp_n$ we exploit symmetries swapping rows to assume without any loss in the generality that the first column starts with either $\frac{n + 1}{2}$ zeros or $\frac{n + 1}{2}$ ones.
This is sufficient to reduce $\bphp_n$ to a smaller formula $\bphp_{n/2}$, which we solve recursively.

The refutation of Tseitin formulas on grids has a more complex structure, resembling the folklore technique of dynamic programming on broken profile.
We reduce the problem on toroidal grids to the problem on the planar grid, where we structure our proof similarly to a dynamic programming solution to the problem of counting the number of satisfying assignments.
It allows us to partition all possible assignments into a constant number of symmetric classes, which leads to the polynomial-size proof.
Although we consider grids, as they produce the best constants, this technique is applicable to Tseitin formulas built on top of any family of graphs that have a lot of small and well connected cycles.

\subparagraph*{Separations.}
For the separations, we transform our basic formulas by embedding extra variables in the formula structure, thereby increasing the size of the smallest ``useful'' symmetry.
Although the size of the symmetries required to efficiently refute $\bphp_n$ is logarithmic in $n$, the transformation that works with this family of formulas allows us to separate systems where the symmetry sizes differ by as little as $1$.
At the same time, Tseitin formulas have the advantage of being polynomially solved using only constant symmetries, but the transformation lifting the lower bounds for local symmetries leads to a multiplicative blow-up in the size of symmetries in the upper bound.
However, this still allow us to separate $\SRCI[k]$ from $\SRCI[(k + 1)]$ and separate bigger global symmetries from smaller local ones, which would not be possible using previous approaches.
Our separations of proof systems with small symmetries from systems with larger symmetries of different type are also based on the Tseitin formulas.
They share the same core idea, but use more complex constructions.

\subsection{Organization of the paper}

In Section~\ref{sec:binary-php}, we consider formulas based on binary encoding of the pigeonhole principle.
We show that Resolution with global logarithmic symmetries admits polynomial-size proofs for the formula $\bphp$, which has exponential lower bounds in Resolution and in stronger systems such as constant-depth Frege.
Then, we use this bound to study the hierarchy of proof systems with symmetries of at least logarithmic size (see Figure~\ref{fig:diagram_bphp}).

Section~\ref{sec:tseitin-formulas} is dedicated to formulas built on top of Tseitin contradictions on two types of grids, rectangular and triangular.
Using these formulas, we show that Resolution with static symmetries may have short proofs for formulas hard for systems like constant-depth Frege even when the bound on the size of symmetries is as small as the constant three.
Then, we use these bounds to obtain the results on the hierarchy of proof systems with symmetries that are smaller than logarithmic (see Figure~\ref{fig:diagram_tseitin_triangular}).
For the sake of simplicity, we construct some of our proofs for the case of rectangular grid first, obtaining worse parameters than it is possible to obtain in the triangular case (see Figure~\ref{fig:diagram_tseitin_rectangular}), and then generalize them to the triangular grid.

In Section~\ref{sec:types}, we consider separations in the directions opposite to those proven in the previous sections and obtain the relations between proof systems with the smallest symmetry constraints (see Figure~\ref{fig:diagram_comb}).

\section{Preliminaries}\label{sec:definitions}

\begin{definition}
    A mapping $\phi$ of literals into literals is a \emph{$k$-symmetry} if and only if it satisfies the following two conditions:
    \begin{enumerate}
        \item $\phi(\lnot x) = \lnot \phi(x).$
        \item All variables except for at most $k$ are mapped into themselves (i.e. $\phi(x) = x$).
    \end{enumerate}

    We call a $k$-symmetry \emph{positive} if it maps variables into variables without negation.
    Symmetries that are not necessarily positive we call \emph{complementary} symmetries.
    In case $k$ is equal to the number of variables, we omit the parameter $k$ and obtain the classical definitions of symmetries~\cite{Krishnamurthy85,Urquhart99}.
\end{definition}

\begin{definition}
    Let $k \geq 1$ be a constant or a function of the number of variables.
    \emph{Resolution with dynamic $k$-symmetries ($\SRCIII[k]$)} is a proof system with the following rules:
    \begin{enumerate}
        \item Resolution rule.
        \[
            \begin{prooftree}
                \hypo{A \lor x}
                \hypo{B \lor \lnot x}
                \infer2{A \lor B}
            \end{prooftree}
        \]
        The variable $x$ is called the \emph{pivot} of the inference.
        \item Symmetry rule.
        \[
            \begin{prooftree}
                \hypo{\phi(B_1), \phi(B_2), \ldots, \phi(B_l)}
                \infer1[$\phi$]{\phi(A)}
            \end{prooftree}
        \]
        where $\phi$ is a $k$-symmetry and $A$ was inferred from $B_1, B_2, \ldots, B_l$ before.
        Additionally, each time this rule is applied, we write down $B_1, B_2, \ldots, B_l$ and $\phi$, for the purposes of verification.
        We call the inference of $A$ from $B_1, B_2, \ldots, B_l$ the justification of the symmetry rule, $\phi(B_1), \phi(B_2), \ldots, \phi(B_l) \vdash \phi(A)$ the result of the symmetry rule, and $\phi(A)$ the resulting clause.
    \end{enumerate}
    Similarly to Resolution, the proof is a sequence of clauses, where each clause is either an input clause or is inferred from previous clauses by one of the rules above, and the last clause is the empty clause.
\end{definition}

The motivation behind the symmetry rule is as follows.
If $A$ was inferred from $B_1, \ldots, B_l$, and there is a $k$-symmetry $\phi$ mapping all premises to already derived clauses, then applying $\phi$ to the proof yields an inference of $\phi(A)$ from $\phi(B_1), \ldots, \phi(B_l)$. Hence, instead of repeating the proof, we may infer $\phi(A)$ in a single step.

\begin{remark}
    It is possible to define $k$-homomorphisms and the corresponding proof systems.
    Moreover, some of results could be generalized to homomorphisms---either directly or with some additional technical work.
    In this work, however, we restrict our attention to symmetries.
\end{remark}

We define Resolution with $k$-symmetries as a system with dynamic symmetries.
That is, the proof is not restricted to symmetries of the original formula and may instead use symmetries that appear after several inference steps, or in a more complex manner.
However, it is easy to define versions of Resolution with $k$-symmetries that use only static (i.e.\ global or local) symmetries.

\begin{definition}
    \emph{Resolution with local $k$-symmetries ($\SRCII[k]$)} is a Resolution with $k$-symmetries with an additional requirement that, in the symmetry rule, all $B_i$ and all $\phi(B_i)$ are present in the input formula.

    \emph{Resolution with global $k$-symmetries ($\SRCI[k]$)} is a Resolution with $k$-symmetries with an additional requirement that, in the symmetry rule, both sets $\set{B_i}$ and $\set{\phi(B_i)}$ are equal to the set of all input clauses.
\end{definition}

If $n$ is the number of variables, $\SRCIII[n]$ is $\SRCIII$ as defined in~\cite{Szeider05} and $\SRCI[n]$ and $\SRCII[n]$ correspond to $\SRCI$ and $\SRCII$ as defined in~\cite{Urquhart99}, respectively.

\begin{definition}
    \emph{$\SRI[k]$, $\SRII[k]$, and $\SRIII[k]$} are the proof systems $\SRCI[k]$, $\SRCII[k]$, and $\SRCIII[k]$, respectively, with an additional requirement for all symmetries to be positive.
\end{definition}

Analogously, if $k$ is set to $n$, we get the original definitions of $\SRI$, $\SRII$, and $\SRIII$~\cite{Krishnamurthy85,Szeider05}.
Now we need to show that all the proof systems above are valid proof systems in the sense of the following definition of Cook and Reckhow~\cite{CR79}.

\begin{definition}
    A \emph{proof system} for CNF is a polynomial-time algorithm $\Pi$ that accepts a pair of a CNF-formula $\Phi$ and a binary string (i.e.\ refutation or proof) $\xi$ and satisfies the following conditions.

    \begin{itemize}
        \item Completeness: for every unsatisfiable $\Phi$ there exists a refutation $\xi$ such that $\Pi(\Phi, \xi) = 1$.

        \item Soundness: if $\Phi$ does not encode an unsatisfiable formula, then $\Pi(\Phi, \xi) = 0$ for any $\xi$.
    \end{itemize}
\end{definition}

\begin{lemma}
    All defined systems are Cook-Reckhow proof systems.
\end{lemma}
\begin{proof}
    The completeness follow from the fact that all systems are extensions of Resolution, which is a Cook-Reckhow proof system.
    Polynomial-time verification is possible since we write down the used symmetry function each time the symmetry rule is applied.
    And the soundness follows from the requirement for the justification to be inferred before the application of the symmetry rule, so the result of the symmetry rule always encodes a valid inference.
\end{proof}

Similarly to Resolution, the size of a proof in systems with symmetries could be measured as both the number of steps and the total length of the proof.
However, those metrics differ only by a polynomial factor, so we will consider only the number of steps.

\begin{definition}
    \emph{$S_{\Pi}(\Phi)$} is the number of steps in the shortest proof of an unsatisfiable formula $\Phi$ in proof system $\Pi$.
\end{definition}

Note that due to the symmetry rule, the size of a proof may be even smaller than the number of clauses in the formula.
An example of such a situation will be given in Theorem~\ref{thrm:bphp_upper}.
Another classical notion we operate with is the notion of p-simulation.

\begin{definition}
    A proof system $A$ \emph{p-simulates} a proof system $B$ if there exists a polynomial-time function $f$ such that $f(\Phi, \xi)$ is a valid $A$-proof of $\Phi$ whenever $\xi$ is a valid $B$-proof of $\Phi$. In other words, $B(\Phi, \xi) = 1 \Rightarrow A(\Phi, f(\Phi, \xi)) = 1$, for all $\Phi, \xi$.

    Two systems simulating each other are called \emph{p-equivalent}.
    Proof system $A$ is \emph{exponentially separated} from $B$ if there exists a family of formulas $\Phi_n$ such that $S_A(\Phi_n) = \BigO(\poly(n))$ and $S_B(\Phi_n) = 2^{\Omega(n)}$.
\end{definition}

\begin{remark}
    A system with small symmetries of a particular type is trivially p-simulated by systems with the larger symmetries of the same type. However, the problem of obtaining separations between such systems is not trivial and will be one of our main goals in this work.
\end{remark}

In our definition, Resolution with $k$-symmetries does not have the weakening rule:
\[
    \begin{prooftree}
        \hypo{A}
        \infer1{A \lor x}
    \end{prooftree}.
\]
Addition of the weakening rule may potentially make Resolution with small dynamic symmetries stronger.
However, in the static case, it was proven by Szeider~\cite{Szeider05} that we can eliminate all applications of the weakening rule similarly to Resolution, and this result can be directly extended to all versions of Resolution with small static symmetries.

\begin{remark}\label{rem:weakening}
    For all $k$, the addition of the weakening rule to $\SRI[k]$, $\SRCI[k]$, $\SRII[k]$, and $\SRCII[k]$ does not change the length of the shortest proof.
\end{remark}

Since we prove only upper bounds on systems with dynamic symmetries, we do not include the weakening rule in the definition.

\begin{lemma}
    $\SRCI[1]$ is p-equivalent to $\Res$.
\end{lemma}
\begin{proof}
    $\SRCI[1]$ works with global symmetries that can flip only one variable.
    Consider an unsatisfiable formula $\Phi$, and let $S$ be the set of all variables corresponding to the $1$-symmetries of $\Phi$.
    By definition of a global symmetry, each clause of $\Phi$ has its copies with all possible signs of variables from $S$.
    Therefore, we can infer, in a polynomial number of resolution steps, a formula $\Psi$ obtained from $\Phi$ by removing all variables from $S$.

    Now consider the $\SRCI[1]$-proof of $\Phi$.
    Repeating it starting from $\Psi$ and ignoring the symmetry rule, we obtain a proof in Resolution with the weakening rule (since some resolution steps in $\SRCI[1]$ may use pivots from $S$, which are no longer present in $\Psi$ and thus correspond to weakening steps).
    Then, we only need to eliminate all applications of the weakening rule to conclude the proof.
\end{proof}

The approach in this lemma cannot be extended to $\SRCII[1]$ because removal of variables may break some symmetries.
Moreover, it is likely that $\SRCII[1]$ is strictly stronger than Resolution.

\section{Binary PHP}\label{sec:binary-php}

In this section, we consider formulas based on the binary pigeonhole principle and prove the relations between proof systems shown in Figure~\ref{fig:diagram_bphp}.

\begin{definition}
    Let $n$ be a power of $2$.
    \emph{$\bphp_n$} is the following CNF-formula in variables $x_{1, 1}, x_{1, 2}, \ldots, x_{1, \log n}, x_{2, 1}, \ldots, x_{n + 1, \log n}$:
    \[\bigwedge_{\substack{i < j \\ 0 \leq y < n}} (\overline{x_{i, 1} \ldots x_{i, \log n}} \neq y) \lor (\overline{x_{j, 1} \ldots x_{j, \log n}} \neq y).\]

    Here $\overline{x_1 x_2 \ldots x_k}$ denotes the integer number whose binary representation is $x_1 x_2 \ldots x_k$ and $\log$ denotes the binary logarithm.
    Note that the non-equalities in the formula above can be expressed by disjunctions, so the formula is in CNF already.
\end{definition}

Informally, $\bphp_n$ encodes the statement that $n + 1$ pigeons cannot be placed into $n$ holes, where each pigeon $i$ is assigned a hole by its binary string $\overline{x_{i,1} \ldots x_{i, \log n}}$.
Each clause encodes the constraint that two distinct pigeons $i$ and $j$ cannot both be mapped to the same hole $y$.

$\bphp_n$ consists of $\BigO(n \log n)$ variables and $\BigO(n^3)$ clauses.
The formula is unsatisfiable and does not change when we swap any two rows of the matrix $(x_{i, j})$.
$\bphp$ is hard not only for Resolution, but also for stronger proof systems.
For example, $\bphp$ is hard for bounded-depth Frege, since the pigeonhole principle with functional and onto axioms ($\ofphp$) is hard for it~\cite{Krajicek95,Pitassi93}, and it is easy to derive in bounded depth $\bphp$ axioms from its axioms.

\begin{theorem}[\cite{Krajicek95,Pitassi93}]
    \(
        S_{\text{depth-$d$ Frege}}(\ofphp_n) = \exp \parens*{n^{\exp(-\BigO(d))}}
    \)
\end{theorem}

\begin{theorem}\label{thrm:bphp_upper}
    \(
        S_{\SRI[2 \log n]}(\bphp_n) = \BigO(n^2)
    \)
\end{theorem}
\begin{proof}
    We only use symmetries to swap rows of the matrix $(x_{i, j})$ in the formula.
    It is clear that each such symmetry step is a valid step in our proof system because the formula itself is symmetrical on swapping rows.

    \begin{figure}
    \centering
    \begin{tabular}{|c|c|c|c|}
        \hline
        \cellcolor{LimeGreen}$x_{1, 1}$ & \cellcolor{cyan}$x_{1, 2}$ & \cellcolor{cyan}$x_{1, 3}$ \\
        \hline
        \cellcolor{LimeGreen}$x_{2, 1}$ & \cellcolor{cyan}$x_{2, 2}$ & \cellcolor{cyan}$x_{2, 3}$ \\
        \hline
        \cellcolor{LimeGreen}$x_{3, 1}$ & \cellcolor{cyan}$x_{3, 2}$ & \cellcolor{cyan}$x_{3, 3}$ \\
        \hline
        \cellcolor{LimeGreen}$x_{4, 1}$ & \cellcolor{cyan}$x_{4, 2}$ & \cellcolor{cyan}$x_{4, 3}$ \\
        \hline
        \cellcolor{LimeGreen}$x_{5, 1}$ & \cellcolor{cyan}$x_{5, 2}$ & \cellcolor{cyan}$x_{5, 3}$ \\
        \hline
        \cellcolor{LimeGreen}$x_{6, 1}$ & $x_{6, 2}$                 & $x_{6, 3}$                 \\
        \hline
        \cellcolor{LimeGreen}$x_{7, 1}$ & $x_{7, 2}$                 & $x_{7, 3}$                 \\
        \hline
        \cellcolor{LimeGreen}$x_{8, 1}$ & $x_{8, 2}$                 & $x_{8, 3}$                 \\
        \hline
        \cellcolor{LimeGreen}$x_{9, 1}$ & $x_{9, 2}$                 & $x_{9, 3}$                 \\
        \hline
    \end{tabular}
    \caption{Variables of $\bphp_8$.}\label{fig:bphp8}
\end{figure}

    Our plan is to infer two clauses stating that the first column of the matrix $(x_{i, j})$ (green part in~Figure~\ref{fig:bphp8}) cannot start with $n / 2 + 1$ zeros nor $n / 2 + 1$ ones and then use symmetries to obtain a contradiction from them.
    The last step is possible because every assignment to the first column contains either at least $n / 2 + 1$ zeros or at least $n / 2 + 1$ ones, and after swapping some rows we can put these values to the first $n / 2 + 1$ rows.
    The construction is recursive.

    First of all, we want to infer clauses $(x_{1, 1} \lor x_{2, 1} \lor \cdots \lor x_{n / 2 + 1, 1})$ and $(\lnot x_{1, 1} \lor \lnot x_{2, 1} \lor \cdots \lor \lnot x_{n / 2 + 1, 1})$.
    In order to do that, we recursively refute the similar $\bphp_{n / 2}$ subformula on variables $x_{1, 2}, x_{1, 3}, \ldots, x_{1, \log n}, x_{2, 2}, \ldots, x_{n / 2 + 1, \log n}$ (cyan part in~Figure~\ref{fig:bphp8}).
    This operation should be done twice, once per each desired clause.
    Then, we construct a contradiction in the first column (green part in Figure~\ref{fig:bphp8}) using these two clauses.
    Consider all partial assignments on the first column.
    We call a partial assignment \textit{sorted} if all variables assigned to zero precede all variables assigned to one.
    Similarly, a clause is \textit{sorted} if and only if it encodes a negation of a sorted assignment.

    It is easy to see that all partial assignments of the size $n + 1$ contradict one of two derived clauses due to symmetry.
    By Remark~\ref{rem:weakening}, we can therefore use the weakening rule to explicitly infer all sorted clauses of size $n + 1$.
    This results in $O(n)$ clauses, each obtained using $O(n)$ symmetry steps, yielding a total size of $O(n^2)$ for this part of the proof.

    Next, each sorted clause of size $k \leq n$ can be derived by a single resolution step applied to two (not necessarily sorted) clauses of size $k+1$, and every unsorted assignment is symmetric to a sorted one. More formally, each sorted clause is derived via the following resolution step (we omit column indices for clarity, since all variables lie in the same first column):
    \[
        \begin{prooftree}
            \hypo{\lnot x_1 \lor \cdots \lor \lnot x_k \lor x_{k + 1} \lor \cdots \lor x_{t + 1}}
            \hypo{\lnot x_1 \lor \cdots \lor \lnot x_k \lor x_{k + 1} \lor \cdots \lor x_t \lor \lnot x_{t + 1}}
            \infer2{\lnot x_1 \lor \cdots \lor \lnot x_k \lor x_{k + 1} \lor \cdots \lor x_t}
        \end{prooftree}.
    \]
    Each unsorted clause in this part of the proof is obtained by applying a symmetry step to the proof of $\lnot x_1 \lor \lnot x_2 \lor \cdots \lor \lnot x_k \lor x_{k + 1} \lor \cdots \lor x_{t} \lor x_{t + 1}$.
    Namely, we apply the symmetry $\phi_{k + 1,t + 1}$ that swaps rows $k + 1$ and $t + 1$ in the matrix $(x_{i,j})$:
    \[
        \begin{prooftree}
            \hypo{\bphp_{n}}
            \infer1[$\phi_{k + 1, t + 1}$]{\lnot x_1 \lor \lnot x_2 \lor \cdots \lor \lnot x_{k + 1} \lor x_{k + 2} \lor \cdots \lor x_t \lor \lnot x_{t + 1}}
        \end{prooftree}.
    \]
    Proceeding in decreasing order of size, we refute all sorted assignments, eventually deriving the empty sorted assignment, which yields a contradiction.

    There are $O(n^2)$ sorted assignments, so the overall size of the proof is $S_{\SRI[2 \log n]}(\bphp_n) = 2S_{\SRI[2 \log n]}\parens*{\bphp_{n / 2}} + O(n^2)$, and hence $S_{\SRI[2 \log n]}(\bphp_n) = O(n^2)$.
\end{proof}

Our next step is, having an upper bound for $\bphp$, construct a formula that separates proof systems with different sizes of allowed symmetries.

\begin{definition}
    Let $n$ be a power of $2$.
    \emph{$\bphp_n^{(t)}$} is the formula obtained by adding extra variables $z_{i, j}$ and $w_{i, j}$ to all clauses of $\bphp_n$ in the following way:
    \[
        (\overline{x_{i, 1} \ldots x_{i, \log n}} \neq y) \lor (\overline{x_{j, 1} \ldots x_{j, \log n}} \neq y) \lor \parens*{\bigvee_{k < t} z_{i, k} \lor z_{j, k}} \lor \parens*{\bigvee_{k < 2 (t + \log n)} w_{y, k}},
    \]
    and also by adding extra clauses that are unit clauses falsifying all the extra variables.
    We call the non-unit clauses main clauses.
\end{definition}

\begin{figure}
    \centering

    \begin{tikzpicture}[auto]
        \node[draw] (r10) [draw = none]                {$\SRI[(2k - 1)]$};
        \node[draw] (r11) [draw = none, above = of r10] {$\SRI[2k]$};

        \node[draw] (rc20) [draw = none, right = 4cm of r10] {$\SRCII[(2k - 1)]$};
        \node[draw] (rc21) [draw = none, above = of rc20] {$\SRCII[2k]$};

        \path[->]
        (rc20) edge [thick] (r10)
        (rc21) edge [thick] (r11)
        ;

        \path[|-]
        (r11) edge [thick] (rc20)
        ;

        \path[|->]
        (r11) edge [thick] (r10)
        (rc21) edge [thick] (rc20)
        ;
    \end{tikzpicture}

    \caption{Relations between proof systems for the case when $k$ is at least logarithmic in number of variables, obtained from Corollary~\ref{cor:bphp_hierarchy}. Arrows with heads indicate p-simulation and arrows with tails indicate exponential separation.}

    \label{fig:diagram_bphp}
\end{figure}
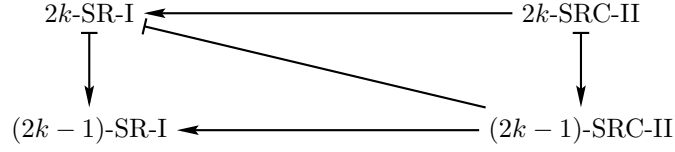

\begin{theorem}\label{thrm:bphp_hierarchy_upper}
    Let $K = 2 (t + \log n)$.
    Then,
    \(
        S_{\SRI[K]}(\bphp_n^{(t)}) = \BigO(n^3 (t + \log n)).
    \)
\end{theorem}
\begin{proof}
    A symmetry swapping two rows $a$ and $b$ is a $K$-symmetry, since it swaps $x_{a, k}$ with $x_{b, k}$ and $z_{a, k}$ with $z_{b, k}$ and no other variables.
    Thus, after propagating all unit-clauses (it takes $\BigO(n^3 (t + \log n))$ steps), we can use the proof constructed in Theorem~\ref{thrm:bphp_upper}.
\end{proof}

For proving the lower bounds for systems with small symmetries, we need the following technical statement.

\begin{theorem}\label{thrm:id_clauses}
    Let $\Pi$ be one of proof systems $\SRI[k]$, $\SRII[k]$, $\SRCI[k]$, or $\SRCII[k]$ for some parameter $k$ (i.e. a proof system with static symmetries).
    Consider a formula $\Phi$ over two sets of variables $X$ and $Y$.
    Suppose that every symmetry of $\Phi$ available in $\Pi$ maps each clause containing variables from $X$ to itself.
    Then, for every substitution $\rho$ to $Y$, $S_\Pi(\Phi) \geq S_\Res(\Phi[\rho])$.
\end{theorem}
\begin{proof}
    Let $C\rvert_X$ denote the set of literals over variables from $X$ in a clause $C$.    
    The key idea is to show that in every application of the symmetry rule with a symmetry $\phi$, each clause $C$ in the justification satisfy the property $C\rvert_X = \phi(C)\rvert_X$.

    We prove this by induction on the structure of the $\Pi$-proof.
    It suffices to verify that this property is preserved by the resolution rule.
    Consider a resolution step $\begin{prooftree} \hypo{A} \hypo{B} \infer2{C} \end{prooftree}$, and assume that $A\rvert_X = \phi(A)\rvert_X$ and $B\rvert_X = \phi(B)\rvert_X$.
    We show that $C\rvert_X = \phi(C)\rvert_X$.
    $C\rvert_X$ and $\phi(C)\rvert_X$ may differ only in the literals associated with the pivot variable.
    This variable could be uniquely determined as the only variable appearing in the premises with different signs.
    Therefore, if it lies in $X$, it is mapped to itself by $\phi$.

    This shows that unwinding the symmetries does not introduce clauses with different sets of literals over $X$ than those that were already present in the original $\Pi$-proof.
    Thus, after doing that and applying substitution $\rho$, we obtain a Resolution-proof of formula $\Phi[\rho]$ of smaller size.
\end{proof}

\begin{theorem}\label{thrm:bphp_hierarchy_lower}
    Let $K = 2 (t + \log n)$.
    Then,
    \(
        S_{\SRCII[(K - 1)]}(\bphp_n^{(t)}) \geq S_{\Res}(\bphp_n).
    \)
\end{theorem}
\begin{proof}
    Variables $w_{i, j}$ ensure that there are no symmetries that map main clauses corresponding to different values of $y$ to each other.
    Therefore, the only symmetries of the main clauses are permutations of rows.
    The smallest nontrivial permutation of rows, which is a swap of two rows, is a $K$-symmetry.
    Hence, there are no nontrivial symmetries of size $K - 1$ on main clauses.
    Applying Theorem~\ref{thrm:id_clauses} concludes the proof.
\end{proof}

Combining Theorems~\ref{thrm:bphp_hierarchy_upper} and~\ref{thrm:bphp_hierarchy_lower} together, we obtain a separation for systems with at least logarithmic symmetries.

\begin{corollary}\label{cor:bphp_hierarchy}
    Let $k = o(N)$ and $k \geq \log N$, where $N$ is the number of variables. Then, $\SRI[2 k]$ is exponentially separated from $\SRCII[(2k - 1)]$.
\end{corollary}

The requirement $k = o(N)$ here emerges from the fact that we need the parameter $n$ of our formula to grow, and $N = \Theta(nk)$.

\section{Tseitin Formulas}\label{sec:tseitin-formulas}

In this section, we consider formulas based on Tseitin contradictions on two types of grids: rectangular and triangular.
For the sake of clarity, we demonstrate our construction on the rectangular grid first, obtaining simpler but more restricted separations (see Figure~\ref{fig:diagram_tseitin_rectangular}) and then generalize them to the triangular grid.

We start with the definition of a Tseitin contradiction.
We follow the notation in~\cite{Haastad20}.

\begin{definition}
    Let $G = \tuple{V, E}$ be a graph and $\alpha$ be a mapping from variables to $\mathbb{F}_2$ such that $\sum_{v \in V} \alpha(v) = 1$.
    A formula on the set of Boolean variables $\map{x_e}{e \in E}$ is called a \emph{Tseitin formula} on graph $G$ if, for all $v$, it encodes the following linear equality over $\mathbb{F}_2$:

    \[
        \sum_{e \ni v} x_e = \alpha(v) \pmod{2}.
    \]

    The encoding of each linear equality is done by adding $\deg(v)$-clauses excluding all $2^{\deg(v) - 1}$ assignments on $\map{x_e}{v \in e}$ that contradict the constraint.

    A Tseitin formula is a \emph{Tseitin contradiction} if $\abs{V}$ is odd and $\alpha(v) = 1$ for all $v$.
\end{definition}

\subsection{Rectangular Grid}\label{sec:rectangular}

\input{static/figs/border.tex}
    
Let $T_n$ be the Tseitin contradiction on a toroidal rectangular grid of the size $n \times n$.
This formula has known lower bounds for the bounded-depth Frege proof system~\cite{Haastad20} (see also~\cite{Haastad22}).

\begin{theorem}[\cite{Haastad20}]\label{thrm:haastad}
    For any sufficiently large odd $n$ and $d \leq \frac{\log n}{59 \log \log n}$, the following holds.
    Any depth-$d$ Frege refutation of $T_n$ requires size $\exp(\Omega(n^{\frac{1}{58(d + 1)}}))$.
\end{theorem}

\begin{theorem}\label{thrm:Tseitin3}
    \(
        S_{\SRCII[3]}(T_n) = \BigO(n^2)
    \)
\end{theorem}
\begin{proof}
    Informally, we refute $T_n$ as follows.
    As a first step, we reduce our problem on the toroidal grid to the problem on a planar grid by resolving everything in one row and one column.
    Then, we solve the problem inductively chipping off one node at the time, row by row, left to right.
    The order of nodes is important as it is crucial that the border of the remaining part always consists of constant number of ``straight edges'' (see Figure~\ref{fig:border}), inside which we can exploit symmetries.
    
    More formally, we start by resolving clauses corresponding to the first row and the first column on the grid.
    The result of such inference will be a clause consisting of variables corresponding to two sets of vertical edges that we call the top and bottom borders and two sets of horizontal edges that we call left and right borders.
    We call such clauses \emph{border assignments} (see example in Figure~\ref{fig:border}).
    Clauses corresponding to vertices enclosed by borders are not used in the inference.
    Thus, we call those vertices \emph{unused}.
    There are exponential number of border assignments.
    However, we are interested only in those that set all edges to zero, except for one edge for each border (the leftmost edges for top and bottom borders, and the topmost edges for left and right borders).
    We call such border assignments \emph{basic clauses}.
    There are only eight basic clauses that could be inferred from the clauses of the first row and the first column, so we infer all of them.
    This can be done in linear number of steps.

    Flipping edges along a U-shaped path (i.e.,\ a path that forms a rectangle without one edge) that starts and ends on a border is a local 3-symmetry of a border assignment, because each used clause has even number of flipped edges.
    An example of such U-shaped path is shown in Figure~\ref{fig:border} by dashed red paint.
    Due to the structure of border assignments, eight basic clauses are sufficient to produce any other border assignment with the same parity using only local 3-symmetries.

    Our next goal is to use the topmost unused row to construct all basic clauses for a border enclosing the same vertices without this row.
    In order to achieve it, we consider broken border assignments, i.e.\ assignments to a set of variables that consist of three regular borders and the top border consisting of vertical edges of two neighboring rows and one horizontal edge as shown by colored edges in Figure~\ref{fig:tseitin-dp}.
    Similarly to regular border assignments, we define basic clauses for broken border assignments by setting all edges to zero, except for one edge per regular border and three edges closest to the horizontal edge on a broken border (colored red in Figure~\ref{fig:tseitin-dp}).
    Again, there are only constant number of basic clauses for each broken border.
    And similarly, we can use local 3-symmetries to infer any broken-border assignment from basic clauses.

    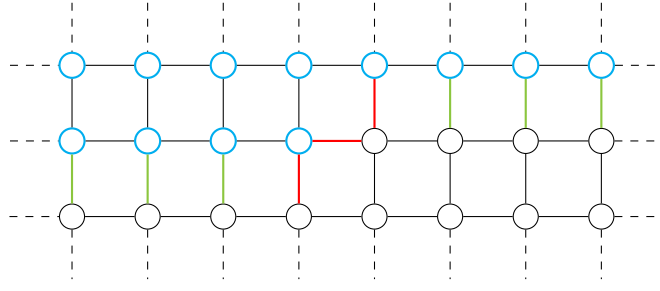
\begin{figure}
    \centering

    \begin{tikzpicture}[auto]
        \node[draw, circle] (00) [draw = none, ] {};
        \node[draw, circle, thick, cyan] (10) [draw = none, below of = 00] {};
        \node[draw, circle, thick, cyan] (20) [draw = none, below of = 10] {};
        \node[draw, circle] (30) [draw = none, below of = 20] {};
        \node[draw, circle] (40) [draw = none, below of = 30] {};
        \node[draw, circle] (01) [draw = none, right of = 00] {};
        \node[draw, circle, thick, cyan] (11) [below of = 01] {};
        \node[draw, circle, thick, cyan] (21) [below of = 11] {};
        \node[draw, circle] (31) [below of = 21] {};
        \node[draw, circle] (41) [draw = none, below of = 31] {};
        \node[draw, circle] (02) [draw = none, right of = 01] {};
        \node[draw, circle, thick, cyan] (12) [below of = 02] {};
        \node[draw, circle, thick, cyan] (22) [below of = 12] {};
        \node[draw, circle] (32) [below of = 22] {};
        \node[draw, circle] (42) [draw = none, below of = 32] {};
        \node[draw, circle] (03) [draw = none, right of = 02] {};
        \node[draw, circle, thick, cyan] (13) [below of = 03] {};
        \node[draw, circle, thick, cyan] (23) [below of = 13] {};
        \node[draw, circle] (33) [below of = 23] {};
        \node[draw, circle] (43) [draw = none, below of = 33] {};
        \node[draw, circle] (04) [draw = none, right of = 03] {};
        \node[draw, circle, thick, cyan] (14) [below of = 04] {};
        \node[draw, circle, thick, cyan] (24) [below of = 14] {};
        \node[draw, circle] (34) [below of = 24] {};
        \node[draw, circle] (44) [draw = none, below of = 34] {};
        \node[draw, circle] (05) [draw = none, right of = 04] {};
        \node[draw, circle, thick, cyan] (15) [below of = 05] {};
        \node[draw, circle] (25) [below of = 15] {};
        \node[draw, circle] (35) [below of = 25] {};
        \node[draw, circle] (45) [draw = none, below of = 35] {};
        \node[draw, circle] (06) [draw = none, right of = 05] {};
        \node[draw, circle, thick, cyan] (16) [below of = 06] {};
        \node[draw, circle] (26) [below of = 16] {};
        \node[draw, circle] (36) [below of = 26] {};
        \node[draw, circle] (46) [draw = none, below of = 36] {};
        \node[draw, circle] (07) [draw = none, right of = 06] {};
        \node[draw, circle, thick, cyan] (17) [below of = 07] {};
        \node[draw, circle] (27) [below of = 17] {};
        \node[draw, circle] (37) [below of = 27] {};
        \node[draw, circle] (47) [draw = none, below of = 37] {};
        \node[draw, circle] (08) [draw = none, right of = 07] {};
        \node[draw, circle, thick, cyan] (18) [below of = 08] {};
        \node[draw, circle] (28) [below of = 18] {};
        \node[draw, circle] (38) [below of = 28] {};
        \node[draw, circle] (48) [draw = none, below of = 38] {};
        \node[draw, circle] (09) [draw = none, right of = 08] {};
        \node[draw, circle, thick, cyan] (19) [draw = none, below of = 09] {};
        \node[draw, circle] (29) [draw = none, below of = 19] {};
        \node[draw, circle] (39) [draw = none, below of = 29] {};
        \node[draw, circle] (49) [draw = none, below of = 39] {};

        \path
(10) edge [dashed] (11)
(01) edge [dashed] (11)
        (20) edge [dashed] (21)
        (11) edge (21)
        (30) edge [dashed] (31)
        (21) edge [thick, LimeGreen] (31)
        (31) edge [dashed] (41)
        (11) edge (12)
        (02) edge [dashed] (12)
        (21) edge (22)
        (12) edge (22)
        (31) edge (32)
        (22) edge [thick, LimeGreen] (32)
        (32) edge [dashed] (42)
        (12) edge (13)
        (03) edge [dashed] (13)
        (22) edge (23)
        (13) edge (23)
        (32) edge (33)
        (23) edge [thick, LimeGreen] (33)
        (33) edge [dashed] (43)
        (13) edge (14)
        (04) edge [dashed] (14)
        (23) edge (24)
        (14) edge (24)
        (33) edge (34)
        (24) edge [thick, red] (34)
        (34) edge [dashed] (44)
        (14) edge (15)
        (05) edge [dashed] (15)
        (24) edge [thick, red] (25)
        (15) edge [thick, red] (25)
        (34) edge (35)
        (25) edge (35)
        (35) edge [dashed] (45)
        (15) edge (16)
        (06) edge [dashed] (16)
        (25) edge (26)
        (16) edge [thick, LimeGreen] (26)
        (35) edge (36)
        (26) edge (36)
        (36) edge [dashed] (46)
        (16) edge (17)
        (07) edge [dashed] (17)
        (26) edge (27)
        (17) edge [thick, LimeGreen] (27)
        (36) edge (37)
        (27) edge (37)
        (37) edge [dashed] (47)
        (17) edge (18)
        (08) edge [dashed] (18)
        (27) edge (28)
        (18) edge [thick, LimeGreen] (28)
        (37) edge (38)
        (28) edge (38)
        (38) edge [dashed] (48)
        (18) edge [dashed] (19)
        (28) edge [dashed] (29)
        (38) edge [dashed] (39)
        ;
    \end{tikzpicture}

    \caption{Representation of a basic clause for a broken border assignment.}

    \label{fig:tseitin-dp}
\end{figure}

    Now, it only remains to observe that all basic clauses of a broken border enclosing $t$ vertices could be obtained from the basic clauses of a broken border enclosing $t + 1$ vertices in a constant number of steps.
    Repeating this process several times, we obtain all basic clauses for the border enclosing a rectangle of the size $1 \times (n - 1)$.
    Then, we can use the similar process to obtain all basic clauses for a single vertex.
    The basic clauses for a single vertex encode the fact that the parity of edges incident to the vertex is even, while clauses of that vertex encode the fact that it is odd.
    Resolving all these clauses we obtain a contradiction.
\end{proof}

\begin{corollary}\label{cor:Rectangular4}
    \(
        S_{\SRCI[4]}(T_n) = \BigO(n^2)
    \)
\end{corollary}
\begin{proof}
    The proof above used only symmetries that flipped variables along the U-shaped paths in the graph.
    Each such path can be extended to a cycle of length $4$ making a local symmetry global.
\end{proof}

Now, having upper bounds for the size of Tseitin formulas on toroidal grid in Resolution with small symmetries, we are ready to construct formulas for the separations. We start with a family of formulas $\eta_n^{(k)}$ separating larger global and local symmetries (with different parameters) from smaller local symmetries. After it, we construct the family of formulas that separate only systems global symmetries, but with much better parameters.

Let $\eta_n^{(k)}$ be a formula $T_n$ with $k$ additional unique variables added to each clause and unit clauses falsifying all additional variables.

\input{static/figs/rectangular-fig}

\begin{theorem}\label{thrm:eta}
    \begin{align}
        S_{\SRCII[(k + 1)]}(\eta_n^{(k)}) & \geq S_{\Res}(T_n) \label{eqn:lower2} \\
        S_{\SRCI[(32 k + 4)]}(\eta_n^{(k)}) & = \BigO(n^2 k) \label{eqn:upper1} \\
        S_{\SRCII[(16 k + 3)]}(\eta_n^{(k)}) & = \BigO(n^2 k) \label{eqn:upper2}
    \end{align}
\end{theorem}
\begin{proof}
    Since $T_n$ has no $1$-symmetries, all $(k + 1)$-symmetries of $\eta_n^{(k)}$ satisfy the conditions of Theorem~\ref{thrm:id_clauses}.
    This gives us the lower bound.

    On the other hand, we can still repeat the proof from Theorem~\ref{thrm:Tseitin3} after using Resolution steps to eliminate all additional variables.
    The only obstacle is to make sure that our system is able to perform the symmetry step that flips two variables along a U-shaped path, only two vertices of which (the ones that are in the base of the letter U) were used in the inference.
    In the case of local symmetries, it is sufficient to add $16k$ additional nontrivially mapped variables to the symmetry, while in the case of global symmetries we need to permute clauses of all four vertices which results in a $(32k + 4)$-symmetry.
\end{proof}

Consider formula $T_n$.
In order to separate $\SRCI[k]$ from $\SRCI[(k + 1)]$, we make the following modifications to this formula.

\begin{itemize}
    \item Split each horizontal edge into $a$ smaller edges and each vertical edge into $b$ smaller edges.
    We refer to all newly added vertices as \emph{small vertices} and to all original vertices as \emph{big vertices}.
    We set the charge of all small vertices to zero to enforce the constraint that each edge is split into smaller edges of the same value.
    \item Add $k$ new variables to each clause corresponding to each big vertex and add unit-clauses falsifying all such variables.
    We refer to these variables as \emph{auxiliary variables}.
\end{itemize}

Let $\xi_{n}^{(a, b, k)}$ be the resulting formula.
Note that the original formula is $\xi_{n}^{(1, 1, 0)}$.
Clearly, $\xi_{n}^{(a, b, k)}$ consists of $\BigO(n^2 (k + a + b))$ variables and clauses.

\begin{figure}
    \centering

    \begin{tikzpicture}[auto]
        \node[draw] (rc11) [draw = none] {$\SRCI[(k + 1)]$};
        \node[draw] (rc12) [draw = none, above = of rc11] {$\SRCI[(k + 2)]$};
        \node[draw] (rc13) [draw = none, above = of rc12] {$\SRCI[(16 k + 3)]$};
        \node[draw] (rc14) [draw = none, above = of rc13] {$\SRCI[(32 k + 4)]$};

        \node[draw] (rc21) [draw = none, right = 4cm of rc11] {$\SRCII[(k + 1)]$};
        \node[draw] (rc22) [draw = none, above = of rc21] {$\SRCII[(k + 2)]$};
        \node[draw] (rc23) [draw = none, above = of rc22] {$\SRCII[(16 k + 3)]$};
        \node[draw] (rc24) [draw = none, above = of rc23] {$\SRCII[(32 k + 4)]$};

        \path[|->]
        (rc12) edge [thick, "Cor~\ref{cor:heirarchy_rectangular} $\parens*{\substack{k \geq 2\\k\text{ is even}}}$"] (rc11)
        (rc13) edge [thick] (rc12)
        (rc14) edge [thick] (rc13)
        ;

        \path[->]

        (rc22) edge [thick] (rc21)
        (rc23) edge [thick] (rc22)
        (rc24) edge [thick] (rc23)

        (rc21) edge [thick] (rc11)
        (rc22) edge [thick] (rc12)
        (rc23) edge [thick] (rc13)
        (rc24) edge [thick] (rc14)
        ;

        \path[|->]
        (rc23) edge [thick, out = -25, in = 25, "Thm~\ref{thrm:eta}"] (rc21)
        ;

        \path[|-]
        (rc14) edge [thick, sloped, "Thm~\ref{thrm:eta}"] (rc21)
        ;
    \end{tikzpicture}

    \caption{Relations between proof systems that follow from the hierarchy theorems for the rectangular grid.}

    \label{fig:diagram_tseitin_rectangular}
\end{figure}
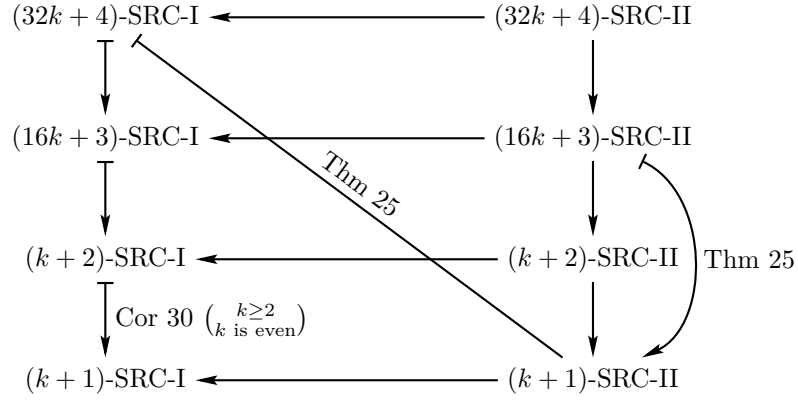

\begin{lemma}
    Let $\phi$ be any global $k$-symmetry of $\xi_{n}^{(a, b, k)}$.
    Then, $\phi(x) \in \set{x, \lnot x}$ for every non-auxiliary variable.
\end{lemma}
\begin{proof}
    All edges adjacent to a big vertex could be mapped only into edges adjacent to the same vertex, otherwise, all auxiliary vertices of this big vertex should be mapped non-trivially, and there are $k$ of them.
    If an edge adjacent to a big vertex is mapped into another edge adjacent to this vertex, then an entire path to some other big vertex starting in that edge should also be mapped to somewhere outside this path.
    This is impossible, since the path ends with an edge adjacent to a big vertex.
    Thus, all edges, adjacent to a big vertex are mapped to either themselves or their negations.
    Therefore, every other non-auxiliary edge is also mapped to either itself or its negation, since it lies on a path between two big vertices.
\end{proof}

\begin{corollary}
    All $k$-symmetries of $\xi_{n}^{(a, b, k)}$ are inversions of sets of variables corresponding to elements of the cycle space of the constructed graph.
\end{corollary}

\begin{theorem}\label{thrm:lower_heirarchy}
    \(
        S_{\SRCI[(2a + 2b - 1)]}(\xi_{n}^{(a, b, a + b)}) \geq S_{\Res}(T_n)
    \)
\end{theorem}
\begin{proof}
    The previous corollary together with the fact that the girth of the constructed graph is $2a + 2b$ show that all $(2a + 2b - 1)$-symmetries work trivially on all non-auxiliary variables.
    Now, we can mark one arbitrary edge on each path and identify each variable with the variable corresponding to the marked edge on its path or its negation.
    After removing all auxiliary variables and clearing all clauses that became trivial after the previous transformation, we obtain a valid Resolution refutation of $T_n$.
\end{proof}

\begin{theorem}\label{thrm:upper_heirarchy}
    \(
        S_{\SRCI[2(a + b)]}(\xi_{n}^{(a, b, a + b)}) = \BigO(n (a + b) + n^2)
    \)
\end{theorem}
\begin{proof}
    We can infer $T_n$ first and then construct the proof, similar to the proof of Theorem~\ref{thrm:Tseitin3}.
    The only difference is that $4$-symmetries from the original proof become $2(a + b)$-symmetries, in which almost all variables (except $4$) never appear in the resulting clause.
\end{proof}

\begin{corollary}\label{cor:heirarchy_rectangular}
    For every even $k \geq 4$, there exists a family of formulas $\Phi_n$ such that the size of the shortest $\SRCI[k]$-proof of $\Phi_n$ is polynomial, but the size of the shortest $\SRCI[(k - 1)]$-proof is exponential.
\end{corollary}
\begin{proof}
    Directly follows from two theorems above.
\end{proof}

\subsection{Triangular Grid}\label{sec:triangular}

\begin{figure}
    \centering

    \begin{tikzpicture}[auto]
        \node[draw] (rc11) [draw = none] {$\SRCI[(k + 1)]$};
        \node[draw] (rc12) [draw = none, above = of rc11] {$\SRCI[(k + 2)]$};
        \node[draw] (rc13) [draw = none, above = of rc12] {$\SRCI[(16 k + 3)]$};
        \node[draw] (rc14) [draw = none, above = of rc13] {$\SRCI[(24 k + 3)]$};

        \node[draw] (rc21) [draw = none, right = 4cm of rc11] {$\SRCII[(k + 1)]$};
        \node[draw] (rc22) [draw = none, above = of rc21] {$\SRCII[(k + 2)]$};
        \node[draw] (rc23) [draw = none, above = of rc22] {$\SRCII[(16 k + 3)]$};
        \node[draw] (rc24) [draw = none, above = of rc23] {$\SRCII[(24 k + 3)]$};

        \path[|->]
        (rc12) edge [thick, "Thm~\ref{thrm:xi_prime} ($k \geq 1$)"] (rc11)
        (rc13) edge [thick] (rc12)
        (rc14) edge [thick] (rc13)
        ;

        \path[->]
        (rc22) edge [thick] (rc21)
        (rc23) edge [thick] (rc22)
        (rc24) edge [thick] (rc23)
        ;

        \path[|->]
        (rc21) edge [thick, "Thm~\ref{thrm:hierarchy_small} ($k \geq 2$)"] (rc11)
        (rc22) edge [gray] (rc12)
        (rc23) edge [gray] (rc13)
        (rc24) edge [thick] (rc14)
        ;

        \path[|->]
        (rc23) edge [thick, out = -25, in = 25, "Thm~\ref{thrm:eta_prime}"] (rc21)
        ;

        \path[|-|]
        (rc14) edge [thick, sloped, "Thm~\ref{thrm:hierarchy_small} ($k \geq 2$)", "Thm~\ref{thrm:eta_prime}" swap] (rc21)
        ;
    \end{tikzpicture}

    \caption{Relations between proof systems that follow from the hierarchy theorems for the triangular grid.}

    \label{fig:diagram_tseitin_triangular}
\end{figure}
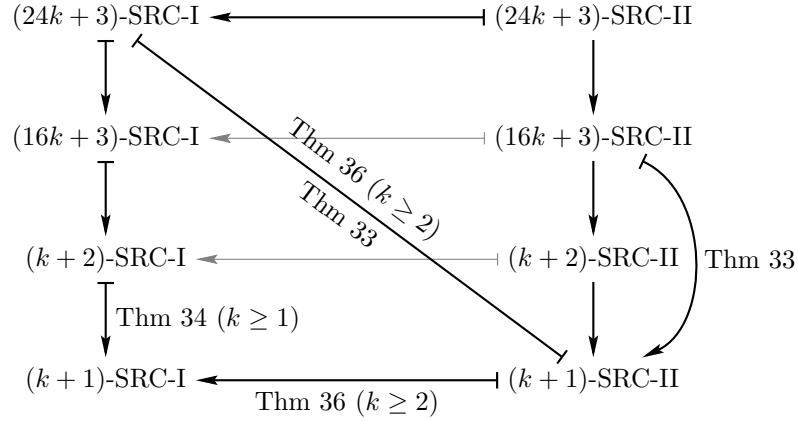

Let $T'_{n}$ be a Tseitin contradiction on a graph constructed as a rectangular toroidal grid with diagonals of the type $(i, j)-(i + 1, j + 1)$ added to each rectangle (see Figure~\ref{fig:triangular}).

\begin{lemma}
    \(
        S_{\text{depth-$d$ Frege}}(T'_n) \geq S_{\text{depth-$d$ Frege}}(T_n).
    \)
\end{lemma}
\begin{proof}
    Setting all diagonals to zero maps $T'_n$ to $T_n$ and any proof of $T'_n$ to a proof of $T_n$.
\end{proof}

\begin{remark}
    The proof of $T'_n$ in Resolution with global symmetries analogous to the proof in Corollary~\ref{cor:Rectangular4} uses $3$-symmetries.
    Thus,
    \(
        S_{\SRCI[3]}(T'_n) = \BigO(n^2).
    \)
\end{remark}

We want to construct a formula similar to $\eta_n^{(k)}$, but with the base formula $T'_n$.
A naive approach---adding $k$ new variables to each clause together with unit clauses falsifying them---leads to worse constants because each vertex constraint in $T'_n$ consists of $32$ clauses, compared to $8$ in $T_n$.
Hence, we need to change the way we add the variables.
For each vertex, we arbitrarily choose $3$ out of $4$ vertical or horizontal edges and split all $32$ clauses into $8$ groups of $4$ clauses with respect to the signs of those three edges.
In each group, we add $k$ new variables to each clause and $k$ unit clauses falsifying them.
Added variables are unique for each group, but they are the same for the clauses inside the group.
As a result, each vertex constraint contains only $8k$ additional variables.
We call the constructed formula $\eta_n'^{(k)}$.

\begin{theorem}\label{thrm:eta_prime}
    \begin{align}
        S_{\SRCII[(k + 1)]}(\eta_n'^{(k)}) & \geq S_{\Res}(T_n)\label{eqn:lower_prime} \\
        S_{\SRCI[(24 k + 3)]}(\eta_n'^{(k)}) & = \BigO(n^2 k)\label{eqn:upper1prime} \\
        S_{\SRCII[(16 k + 3)]}(\eta_n'^{(k)}) & = \BigO(n^2 k)\label{eqn:upper2prime}
    \end{align}
\end{theorem}
\begin{proof}
    Proofs of the upper bounds are similar to the proofs of the upper bounds in Theorem~\ref{thrm:eta}.
    The only difference is that the length of a cycle in a new graph is $3$ instead of $4$ and the analogue of a U-shaped path is either a path of length $2$ or a path of length $3$ which is also a cycle.

    The proof of the lower bound is more interesting.
    Formula $\eta_n'^{(k)}$ does have nontrivial $(k + 1)$-symmetries.
    In fact, it even has $2$-symmetries that can negate some pairs of two neighboring edges.
    However, if we set all diagonal edges to $0$ we obtain the formula $\eta_n^{(k)}$ that does not have $(k + 1)$-symmetries that operate nontrivially on the edges of the rectangular grid.
    Therefore, we still can apply Theorem~\ref{thrm:id_clauses} and obtain the lower bound of $S_{\Res}(T_n)$.
\end{proof}

\input{static/figs/triangular-fig}

Let $\xi_n^{(a, b, c, k)}$ be the formula constructed similarly to $\xi_n^{(a, b, k)}$, but with formula $T'_n$ used as a base formula, where $c$ is the number of parts each diagonal edge is split into and all other parameters having the same meaning as before.

\begin{theorem}\label{thrm:xi_prime}
    For every $k \geq 3$, there exists a family of formulas $\Phi_n$ such that the size of the shortest $\SRCI[k]$-proof of $\Phi_n$ is polynomial, but the size of the shortest $\SRCI[(k - 1)]$-proof is exponential.
\end{theorem}
\begin{proof}
    The formula separating $\SRCI[k]$ and $\SRCI[(k + 1)]$ is $\xi_n^{(a, b, c, a + b + c)}$, where $a$, $b$ and $c$ are chosen in such a way that they satisfy the triangle inequality and $a + b + c = k + 1$.
    The lower bound is obtained similarly to Theorem~\ref{thrm:lower_heirarchy} with the only difference that the girth of the new graph is $a + b + c$.

    \[
        S_{\SRCI[(a + b + c - 1)]}(\xi_n^{(a, b, c, a + b + c)}) \geq S_{\Res}(T'_n).
    \]

    And the upper bound is similar to Theorem~\ref{thrm:upper_heirarchy}.

    \[
        S_{\SRCI[(a + b + c)]}(\xi_n^{(a, b, c, a + b + c)}) = \BigO(n(a + b + c) + n^2).
    \]

    Combining them together we obtain the desired separation.
\end{proof}

\section{Separating Small Symmetries from Large Symmetries}\label{sec:types}

There are three possible metrics on which one system with small symmetries can have an advantage over another.
Each of these metrics leads to a different type of separation:

\begin{enumerate}
    \item Increasing the size of the symmetries. An example of such statement is: for appropriate values of $k$, $\SRI[k + 1]$ has polynomial proofs for formulas hard for $\SRCII[k]$.
    \item Generalizing the type of the symmetries. For example, $\SRII[6]$ has polynomial proofs for formulas hard for $\SRCI$ (Theorem~\ref{thrm:hierarchy_small}).
    \item Adding complementary symmetries. For example, $\SRCI[164]$ has polynomial proofs for formulas hard for $\SRII$ (Theorem~\ref{thrm:last2}).
\end{enumerate}

In previous sections, we focused on the separations of the first type.
In this section, we consider the other two types, where the system with small symmetries has an advantage over the system with large symmetries.

We start with the separations of small symmetries of more general type from large symmetries of less general type.
In order to obtain them, we need a family of formulas $T'_n$ from the previous section and the following result proven by Szeider.

\begin{theorem}[\cite{Szeider05}]
    \label{thrm:break_all}
    For any formula $\Phi_n$ there exists a formula $\Psi_n$ satisfying the following conditions:

    \begin{itemize}
        \item $\Psi_n$ is polynomially constructable from $\Phi_n$.
        \item $\Phi_n$ could be inferred from $\Psi_n$ in a polynomial number of Resolution steps.
        \item $S_{\SRCII}(\Psi_n) \geq S_{\Res}(\Phi_n).$
    \end{itemize}
\end{theorem}

\begin{theorem}\label{thrm:hierarchy_small}
    There exist families of formulas $\Phi_n$, $\Psi_n$, $\Phi_n'$, $\Psi_n'$ such that:
    \begin{itemize}
        \item $S_{\SRCIII[3]}(\Phi_n)$ is polynomial, but $S_{\SRCII}(\Phi_n)$ is exponential.
        \item $S_{\SRCII[3]}(\Psi_n)$ is polynomial, but $S_{\SRCI}(\Psi_n)$ is exponential.
        \item $S_{\SRIII[6]}(\Phi'_n)$ is polynomial, but $S_{\SRCII}(\Phi'_n)$ is exponential.
        \item $S_{\SRII[6]}(\Psi'_n)$ is polynomial, but $S_{\SRCI}(\Psi'_n)$ is exponential.
    \end{itemize}
\end{theorem}
\begin{proof}
    $\Phi_n$ could be built as $T'_n$ after transformation described in Theorem~\ref{thrm:break_all}.
    We build $\Psi_n$ from $T'_n$ by adding, for each variable of the original formula, a dummy clause whose length is unique among all clauses (including those of $T'_n$), and which contains exactly that variable together with a positive number of new variables.
    The global symmetries of the constructed formula map all original variables trivially into themselves, and the dummy clauses cannot participate in the proof because each of them contain a variable that appear in the formula only with one sign, so all $\SRCI$-proofs of $\Psi_n$ are essentially $\Res$-proofs of $T'_n$.
    On the other hand, all local symmetries are unaffected by the transformation.
    
    The remaining formulas $\Phi'_n$ and $\Psi'_n$ are built the same way as $\Phi_n$ and $\Psi_n$, but with additional transformation in the middle that makes negations of variables new variables and adds constraints ensuring that a variable is never equal to the new variable representing its negation.
    This transformation makes each $k$-symmetry a positive $2k$-symmetry.
\end{proof}

Combining this theorem with Theorem~\ref{thrm:eta_prime} and Corollary~\ref{cor:bphp_hierarchy}, we can obtain all separations shown in Figure~\ref{fig:diagram_comb}, except for the half of the top left edge and the separation between $\SRII$ and $\SRCI[164]$.
\begin{figure}
    \centering

    \begin{tikzpicture}[auto]
        \node[draw] (res) [draw = none] {$\Res$};
        \node[draw] (rc1) [draw = none, right = 2cm of res] {$\SRCI$};
        \node[draw] (rc1s) [draw = none, above = of rc1] {$\SRCI[3]$};
        \node[draw] (r1s) [draw = none, below = of rc1] {$\SRI[6]$};
        \node[draw] (rc2) [draw = none, right = 2cm of rc1] {$\SRCII$};
        \node[draw] (rc2s) [draw = none, above = of rc2] {$\SRCII[3]$};
        \node[draw] (r2s) [draw = none, below = of rc2] {$\SRII[6]$};
        \node[draw] (rc3) [draw = none, right = 2cm of rc2] {$\SRCIII$};
        \node[draw] (rc3s) [draw = none, above = of rc3] {$\SRCIII[3]$};
        \node[draw] (r3s) [draw = none, below = of rc3] {$\SRIII[6]$};
        \node[draw] (r1) [draw = none, above = of res] {$\SRI$};

        \node[draw] (rc1m) [draw = none, right = 1cm of rc3] {$\SRCI[164]$};
        \node[draw] (r2) [draw = none, above = of rc1m] {$\SRII$};

        \path[|->]
        (rc1s) edge [thick] (res)
        (rc1) edge [thick] (res)
        (r1s) edge [thick] (res)

        (rc1) edge [thick] (rc1s)
        (rc1) edge [thick] (r1s)

        (rc2) edge [thick] (rc2s)
        (rc2) edge [thick] (r2s)

        (rc2s) edge [thick] (rc1s)
        (rc3s) edge [thick] (rc2s)

        (r2s) edge [thick] (r1s)
        (r3s) edge [thick] (r2s)

        (rc2) edge [thick] (rc1)
        (rc3) edge [thick] (rc2)

        ;

        \path[->]
        (rc3) edge [thick] (rc3s)
        (rc3) edge [thick] (r3s)
        ;

        \path[|-|]
        (rc2s) edge [thick, sloped, "Thm~\ref{thrm:hierarchy_small}", "Thm~\ref{thrm:eta_prime}" swap] (rc1)
        (r2s) edge [thick] (rc1)
        (r1) edge [thick, "Lem~\ref{lemma:last}", "Cor~\ref{cor:bphp_hierarchy}" swap] (rc1s)
        ;

        \path[|-]
        (rc3s) edge [thick, sloped, "Thm~\ref{thrm:hierarchy_small}"] (rc2)
        (r3s) edge [thick] (rc2)
        ;

        \path[|-|]
        (r2) edge [thick, "Thm~\ref{thrm:last2}"] (rc1m)
        ;
    \end{tikzpicture}

    \caption{Relation between proof systems with the smallest symmetry constraints.}

    \label{fig:diagram_comb}
\end{figure}
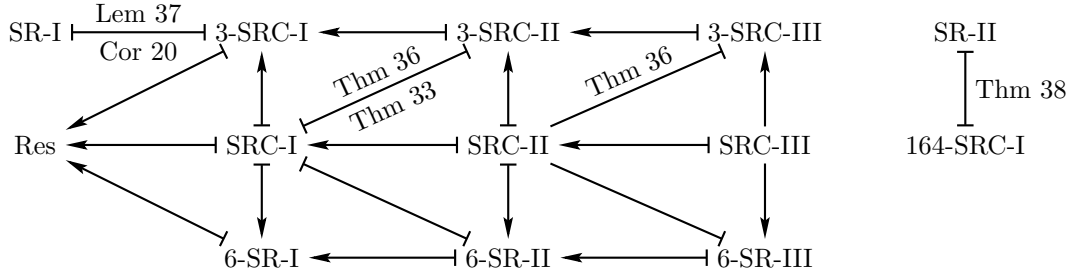%
We show the remaining separations by generalizing the idea of the $\SRCII$ lower bounds of Arai and Urquhart~\cite{AU00}.

\begin{lemma}\label{lemma:last}
    There exists a family of formulas $\tau_n$ such that the size of the shortest $\SRCI[3]$-proof of $\tau_n$ is polynomial, but the size of the shortest $\SRI$-proof is exponential.
\end{lemma}
\begin{proof}
    We modify the formula $T'_n$.
    Enumerate all vertices in an arbitrary order and consider a vertex $v_j$.
    The constraint corresponding to $v_j$ is encoded using $32$ clauses of $6$ literals.
    We add $j$ additional variables $z_{j, 1}, y_{j, 2}, \ldots, y_{j, j}$ to all $32$ clauses and add $j$ unit-clauses falsifying them.
    We call the resulting formula $\tau_n$.

    The proof of $\tau_n$ in $\SRCI[3]$ could be constructed similarly to the proof of $T'_n$ in $\SRCI[3]$ because $\tau_n$ still contains all the symmetries flipping variables along the cycles of length $3$.
    On the other hand, the set of lengths of clauses is unique to each variable of $T'_n$, which means that all global positive symmetries work trivially on them.
    Combination of Theorem~\ref{thrm:id_clauses} with the Resolution lower bound for $T'_n$ concludes the proof.
\end{proof}

Although the constructed formula $\tau_n$ does not have any nontrivial global positive symmetries, it still has local symmetries. This can be avoided by using more complicated transformation, which we show in the next theorem.

\begin{theorem}\label{thrm:last2}
    There exists a family of formulas $\tau'_n$ such that the size of the shortest $\SRCI[164]$-proof of $\tau'_n$ is polynomial, but the size of the shortest $\SRII$-proof is exponential.
\end{theorem}
\begin{proof}
    Unlike the previous lemma, we base this construction on the family of formulas $T_n$ built on the rectangular (instead of the triangular) toroidal grid.
    Again, as the first step, we enumerate all vertices and edges incident to each vertex in an arbitrary order.
    Let $x_1$, $x_2$, $x_3$, and $x_4$ be the variables from $8$ clauses corresponding to the $j$-th vertex.
    We replace these eight clauses with eight new clauses of $j + 9$ variables each:
    \begin{align*}
    (\phantom{\lnot} x_{1} \lor \phantom{\lnot} x_{2} \lor \phantom{\lnot} x_{3} \lor \phantom{\lnot} x_{4} &&\lor&& \phantom{\lnot} y_{1, 1} \lor \phantom{\lnot} y_{1, 2} \lor \phantom{\lnot} y_{1, 3} \lor \phantom{\lnot} y_{1, 4} \lor \phantom{\lnot} y_{1, 5} &&\lor&& z_1 \lor \ldots \lor z_j) \\
    (\lnot x_{1} \lor \lnot x_{2} \lor \phantom{\lnot} x_{3} \lor \phantom{\lnot} x_{4} &&\lor&& y_{2, 1} \lor \phantom{\lnot} y_{2, 2} \lor \phantom{\lnot} y_{2, 3} \lor \phantom{\lnot} y_{2, 4} \lor \phantom{\lnot} y_{2, 5} &&\lor&& z_1 \lor \ldots \lor z_j) \\
    (\lnot x_{1} \lor \phantom{\lnot} x_{2} \lor \lnot x_{3} \lor \phantom{\lnot} x_{4} &&\lor&& \lnot y_{3, 1} \lor \phantom{\lnot} y_{3, 2} \lor \phantom{\lnot} y_{3, 3} \lor \phantom{\lnot} y_{3, 4} \lor \phantom{\lnot} y_{3, 5} &&\lor&& z_1 \lor \ldots \lor z_j) \\
    (\lnot x_{1} \lor \phantom{\lnot} x_{2} \lor \phantom{\lnot} x_{3} \lor \lnot x_{4} &&\lor&& \lnot y_{4, 1} \lor \lnot y_{4, 2} \lor \phantom{\lnot} y_{4, 3} \lor \phantom{\lnot} y_{4, 4} \lor \phantom{\lnot} y_{4, 5} &&\lor&& z_1 \lor \ldots \lor z_j) \\
    (\phantom{\lnot} x_{1} \lor \lnot x_{2} \lor \lnot x_{3} \lor \phantom{\lnot} x_{4} &&\lor&& \lnot y_{5, 1} \lor \lnot y_{5, 2} \lor \lnot y_{5, 3} \lor \phantom{\lnot} y_{5, 4} \lor \phantom{\lnot} y_{5, 5} &&\lor&& z_1 \lor \ldots \lor z_j) \\
    (\phantom{\lnot} x_{1} \lor \lnot x_{2} \lor \phantom{\lnot} x_{3} \lor \lnot x_{4} &&\lor&& \lnot y_{6, 1} \lor \lnot y_{6, 2} \lor \lnot y_{6, 3} \lor \lnot y_{6, 4} \lor \phantom{\lnot} y_{6, 5} &&\lor&& z_1 \lor \ldots \lor z_j) \\
    (\phantom{\lnot} x_{1} \lor \phantom{\lnot} x_{2} \lor \lnot x_{3} \lor \lnot x_{4} &&\lor&& \lnot y_{7, 1} \lor \lnot y_{7, 2} \lor \lnot y_{7, 3} \lor \lnot y_{7, 4} \lor \lnot y_{7, 5} &&\lor&& z_1 \lor \ldots \lor z_j) \\
    (\lnot x_{1} \lor \lnot x_{2} \lor \lnot x_{3} \lor \lnot x_{4} &&\lor&& \lnot y_{8, 1} \lor \lnot y_{8, 2} \lor \lnot y_{8, 3} \lor \lnot y_{8, 4} \lor \lnot y_{8, 5} &&\lor&& z_1 \lor \ldots \lor z_j)
\end{align*}
    where the new variables $y$ and $z$ are unique for each vertex and do not appear in any other clauses.
    Eight new clauses are constructed to satisfy the following two properties:
    \begin{itemize}
        \item All eight clauses have the same length, and this length is unique to the vertex.
        \item The numbers of negated literals are distinct in all eight clauses.
    \end{itemize}
    Finally, we add unit clauses for the new variables so that the corresponding literals in the other clauses get falsified, and we call the resulting formula $\tau'_n$.

    To prove the separation, we need to show that $S_{\SRCI[164]}(\tau'_n)$ is polynomial and $S_{\SRII}(\tau'_n) \geq S_{\Res}(T_n)$.
    For the first part, it is sufficient to show that symmetries used to construct the $\SRCI[4]$ proof of $T_n$ could be extended to $164$-symmetries of $\tau'_n$.
    Since such symmetries are just flipping variables along cycles of length $4$, we only need to show that the operation of flipping two variables $x_{e_{j, a}}$ and $x_{e_{j, b}}$ could be correctly extended to a symmetry of all eight new clauses of $j$-th vertex.
    Indeed, all variables $z$ can be trivially mapped to themselves, and the mapping on $y$ can be defined straightforwardly because all $40$ such variables are distinct.
    The total size of such symmetry is $164$ because we need to flip $4$ variables $x$, nontrivially map $40$ variables $y$ per each of $4$ vertices in the cycle.

    For the lower bound, we need to show that all local positive symmetries in $\tau'_n$ map clauses containing original variables $x$ to themselves.
    This is because the pair of the length of a clause and the number of negations in it is unique to each clause with original variables $x$, and both these metrics are invariant under positive symmetries.
    Thus, we can apply Theorem~\ref{thrm:id_clauses}.
\end{proof}

\begin{remark}
    The number $164$ in the previous theorem is not optimal, and it can be reduced by a more careful construction.
    For example, replacing the rectangular grid with a hexagonal grid (dual to the graph of $T'_n$) would reduce this number to $54$ as we would need only $8$ variables $y$ per vertex.
    However, doing this would require a lot of technical work, and it would not add much to the understanding of the problem, so we do not pursue this direction.
\end{remark}

As an immediate corollary, we obtain the following separation that was not known before.

\begin{corollary}\label{cor:last2}
    There exists a family of formulas $\tau'_n$ such that the size of the shortest $\SRCI$-proof of $\tau'_n$ is polynomial, but the size of the shortest $\SRII$-proof is exponential.
\end{corollary}

\section{Conclusion and Open Problems}\label{sec:open}

We have shown that the systems $\SRCI[3]$ and $\SRI[6]$ are stronger than Resolution.
The question of whether systems with smaller constraints are stronger than Resolution is open.
It seems plausible that a separation exists even for $\SRCII[1]$ and Resolution.
Separating $\SRCIII[2]$ from Resolution would also be of independent interest: while $2$-symmetries may efficiently encode parities of some variables, an efficient decoding of those parities from a polynomial-size formula often requires larger symmetries.

For proof systems with global symmetries and for systems with at least logarithmic symmetries, we have shown that the addition of a constant to the size of allowed symmetries makes the system stronger.
However, for systems with sublogarithmic local symmetries, we only have a separation result with respect to a multiplicative factor, but not an additive constant.
It seems natural for this case to also have the separation for the additive constant, so it would be interesting to find it.

Although in this paper we focused on static symmetries, it would be interesting to study the case of small dynamic symmetries as well.
Even the simplest case of dynamic symmetries, $\SRCIII[1]$, seems to be very nontrivial.
1-symmetries grant the system the power to express, although with some restrictions, conjunctions of sets of variables.
For instance, $(x \land y)$ could be expressed as a \emph{single} clause $(x \lor y)$ with two symmetries negating $x$ and $y$.
It raises the potential power of the system to $\Res(k)$ (a version of the Resolution proof system that works with k-DNFs instead of clauses).
However, due to the structure of the proof with symmetries, it is not possible to operate on the conjunctions freely, so both upper and lower bounds for $\SRCIII[1]$ are interesting open problems.

Additionally, the notion of small symmetries could also be naturally applied to systems that are stronger than Resolution.
Examples of such systems are: Cutting Planes, $\Res(k)$, or even constant-depth Frege.
It looks hard to prove lower bounds for these systems augmented with large symmetries, but proving lower bounds for systems with small symmetries only seems more feasible.

\section{Acknowledgements}\label{sec:acknowledgements}

The author is very grateful to Edward~A.~Hirsch and Ofer Strichman, who supervised this work, in particular, introduced the author to the use of symmetries in proof search, guided him through numerous discussions, and gave invaluable comments on the presentation of the results.
The author also thanks Arthur Riazanov for fruitful discussions.

\bibliography{main}

\end{document}